\documentclass{iopart}

\usepackage{graphicx}
\usepackage{subfigure}
\usepackage{amsthm}

\newtheorem{theorem}{Theorem}

\begin{document}

\title{Stochastic Pooling Networks}

\author{Mark D. McDonnell}
\address{ Institute for Telecommunications
Research, University of South Australia, Mawson Lakes, SA 5095, Australia }
 \ead{mark.mcdonnell@unisa.edu.au}

\author{Pierre-Olivier Amblard}
\address{GIPSA-lab, Dept. Images and Signals, UMR CNRS 5216, ENSIEG-BP 46, 38402 Saint Martin d'H\`{e}res Cedex, France}
 \ead{bidou.amblard@gipsa-lab.inpg.fr}

\author{Nigel G. Stocks}

\address{School of Engineering, The University of Warwick, Coventry CV4 7AL, UK}
 \ead{ngs@eng.warwick.ac.uk}

\begin{abstract}
We introduce and define the concept of a {\em stochastic
pooling network} (SPN), as a model for sensor systems where redundancy and two forms of `noise' -- {\em lossy compression} and {\em randomness} -- interact in surprising ways. Our approach to analyzing SPNs is information theoretic. We define an SPN as a network with multiple nodes that each produce noisy and compressed measurements of the same information. An SPN must combine all these measurements into a single further compressed network output, in a way dictated solely by naturally occurring physical properties -- i.e. {\em pooling} -- and yet causes no (or negligible) reduction in mutual information. This means SPNs exhibit redundancy reduction as an emergent property of pooling. The SPN concept is applicable to examples in biological neural coding, nano-electronics, distributed sensor networks, digital beamforming arrays, image processing, multiaccess communication networks and social networks.  In most cases the randomness is assumed to be unavoidably present rather than deliberately introduced. We illustrate the central properties of SPNs for several case studies, where pooling occurs by summation, including nodes that are noisy scalar quantizers, and nodes with conditionally Poisson statistics.  Other emergent properties of SPNs and some unsolved problems are also briefly discussed.
\end{abstract}

\maketitle

\section{Introduction and Background}\label{S:Intro}

Two challenges are ubiquitous for many forms of signal and
information processing tasks, whether in biology, or artificial technology. These are (i) robustness to the
effects of random noise, and (ii) efficient extraction of only the
`information' relevant for achieving some goal. The latter challenge
can require {\em lossy compression}, where some `information' is intentionally discarded. Robustness to unavoidable random noise or fluctuations is often achieved using a network or array of sensors. It is less obvious that a network approach can simultaneously achieve both useful lossy compression and noise reduction.

Lossy compression, far from being a limitation, is usually very desirable, because the lost information is deemed either redundant or irrelevant~\cite{Berger.98}. For example, when a histogram is formed from real-valued data, information about individual measurements is discarded in order to drawn conclusions about general trends. Any two measurements that fall within the same histogram bin are treated as equivalent even though the original measurements may have been different. As a further example, the utility of lossy compression techniques based on what is perceptible by our senses, such as the JPEG and MP3 standards, is nowadays self-evident.

The goal of this paper is to introduce and define a concept we call a {\em stochastic pooling network} (SPN). This name was first proposed in~\cite{Zozor.07} to describe a network of sensors where each sensor produces binary measurements of a common information source. In the process of {\em pooling}, this network simultaneously reduces random fluctuations -- i.e. {\em noise} -- via an averaging-like effect, and further compresses the data. The signal processing goal in~\cite{Zozor.07} was a binary detection problem.

The concept described in the current paper evolved from previous work on a model~\cite{Stocks.Mar2000} that can now be described as a special case of an SPN. That model was studied in the context of a form of stochastic resonance~\cite{Wiesenfeld.95,Gammaitoni.Jan98,Bulsara.05} known as suprathreshold stochastic resonance. Nowadays stochastic resonance is a generic term that describes any phenomenon where random noise in a nonlinear system can provide a signal processing benefit~\cite{Kosko,McDonnell}.  The simple SPN in~\cite{Stocks.Mar2000} exhibits stochastic resonance in a much more pronounced way than conventionally is the case, in that `noise benefits' occur for suprathreshold signals and do not rely on small input signal-to-noise ratios (SNR)~\cite{Stocks.Mar2000}. The same model is well suited for studying the information coding properties of populations of parallel neurons~\cite{Stocks.01c,Stocks.02,Hoch.03a}.

We emphasize that here we aim to define {\em Stochastic Pooling Network} in a way that extends its scope significantly beyond that of~\cite{Stocks.Mar2000,Zozor.07}. In particular, our definition means the nodes in the network can be far more complex than the simple quantizers considered in~\cite{Stocks.Mar2000,Zozor.07}, and in related studies on suprathreshold stochastic resonance~\cite{McDonnell.02MEJ,Hoch.03a,Rousseau.03b,McDonnell.07,Oliaei.03,Nguyen.06}.

Here we are also not focussed on suprathreshold stochastic resonance, but emphasize instead the following essential features that a system must posses in order to fit our definition of an SPN:
\begin{itemize}
  \item multiple noise sources -- usually unavoidable and uncontrollable -- corrupting multiple measurements of the same sample of an `information source';
  \item lossy compression of each noisy measurement;
  \item `pooling' of the multiple noisy and compressed measurements to a single measurement that has fewer states than the vector of individual measurements.
\end{itemize}

By {\em pooling}, we mean that the multiple measurements of the `information source' are constrained to being combined in a simple way that is not controllable or reversible, and that loses the details of which measurement originated from each measurement source. Below we use the term {\em pooling function} when mathematically describing the precise manner in which measurements are pooled -- see Section~\ref{S:Define}.

Our motivation for defining an SPN originated from research into two fundamental scientific questions about biological neurons:
\begin{enumerate}
  \item {\em What mechanisms do biological neural systems use to compress information about external stimuli during transduction at the sensory periphery?}
  \item {\em Do unpredictable fluctuations in neural activity in sensory transduction processes contribute to coding/compression effectiveness? If so, is this achieved in conjunction with redundancy?}
\end{enumerate}

The three goals of this paper are to (i) define {\em stochastic pooling network}; (ii) illustrate that SPNs are a generically useful paradigm for a diverse range of scenarios beyond biological neural coding; and (iii) illustrate that SPNs can display surprising emergent behaviour that is closely related to their capabilities for noise reduction and compression, and are therefore interesting to study in their own right.

We proceed in Section~\ref{S:Models} by qualitatively describing a wide range of systems, ranging from the nanoscale to the macroscopic, that illustrate the diversity of the SPN model.  Next, in Section~\ref{S:Define} we state our definition of an SPN. In Section~\ref{S:Cases} we present some simple SPN case studies where we use information theory to explore some of the consequences of our definition. Finally, in Section~\ref{S:Emergent} we discuss several general emergent phenomena that can be expected to occur in any SPN, list some unsolved problems for SPNs, and present some suggestions for further work on the topic.

\section{Examples of Stochastic Pooling Networks}\label{S:Models}

The term {\em node} is used to refer to each of $N$ independently random measurements in an SPN. Note that in some of our discussion we loosely use the word `sensor' to describe the nodes of an SPN. This may be slightly misleading for some scenarios, such as neural populations, where each node is a parallel nerve fibre.  In such circumstances, the word `channel' may be more appropriate.

A simple thought experiment illustrates the SPN principles outlined in Section~\ref{S:Intro}. Suppose a psychophysicist asks $100$ people to vote `yes' or `no' on whether they think a sound is `loud.' After listening to the sound, each person writes down their vote and places it in a hat.  The convenor then counts the votes for `yes,' and obtains a measurement between 0 and 100.  If the experiment is repeated for many randomly chosen sound volumes, the vote counts provide noisy and compressed estimates of each sound volume.  They are compressed because an estimate of an analogue volume is quantized to a discrete scale. They are noisy, because repeated presentation of the same volume may result in a different count. Other features of this thought experiment include (i) multiple noisy measurements of the same information source -- each person may perceive a sound with different biases; (ii) extremely lossy compression in each measurement -- each person is forced to vote on a binary scale; (iii) pooling of measurements -- the convenor doesn't know who votes yes. The examples below outline how the same ideas apply to various data acquisition and processing systems.

\subsection{Artificial Sensor and Communications Networks}

A number of modern sensor and communications networks can be modelled as an SPN, mostly due to some constraint on the network topology, such as limited power in each sensor, or the simplicity inherent in having multiple identical nodes~\cite{Gastpar.05}. An abstraction of such a problem in distributed source coding is known as the Chief Executive Officer (CEO) problem~\cite{Berger.96}. However, work in the signal processing literature is usually concerned with finding optimal designs for aspects of such networks. Our definition of an SPN instead includes a requirement that coding and pooling occurs naturally due to physical properties, and we aim to study the consequences of this, and whether optimal results can emerge. Distributed detection schemes that utilize multiaccess channels~\cite{Liu.07} is a good example of an artificial network where this constraint may be realistic. Indeed, some studies of the latter case~\cite{Li.08} can be easily mapped to the binary detection SPN~\cite{Zozor.07}. One of the earliest digital sonar beamforming systems, known as the DIMUS sonar array, can also be described in this way~\cite{Remley.66}, as well as more recent sonar systems like the Barra sonobuoy.

\subsection{Biological Neurons}

The emergent properties of the simple network considered in~\cite{Stocks.Mar2000}, such as suprathreshold stochastic resonance, have been observed for its extension to networks of neuron models, including the Fitzhugh-Nagumo~\cite{Stocks.01c}, and Hodgkin-Huxley models~\cite{Hoch.03a}. Suprathreshold stochastic resonance has also been observed when replacing additive noise in each node by multiplicative noise~\cite{Nikitin.07}. The group of cochlear nerve fibres that synapse with a single hair cell in the inner ear, and transduce analogue sound waveforms, is an excellent candidate for description as an SPN. This is because multiple parallel nerve fibres code the same signal using action potentials, but have conditionally independent variability~\cite{Johnson.76}.  It is also possible that modulation of synaptic neurotransmitters can be described as an SPN~\cite{Koch}. Such questions are now being addressed in experiments on living neural networks~\cite{Villard.07FaN}.

We expect that studies of SPN models may be useful for the design of future neural prosthetics, such as cochlear implants. These are surgically implanted biomedical prosthetics that allow profoundly deaf people to hear, via electrical stimulation of the cochlear nerve~\cite{Clark}. This electrical stimulation does not replicate the independent variability of healthy cochlear nerve fibres, and many nerve fibres are therefore redundant. It has been proposed that cochlear implants may be improved by the controlled introduction of suprathreshold stochastic resonance~\cite{Stocks.02,Morse.07}. The SPN model is a useful way of understanding why cochlear nerve fibres are naturally independently noisy, and why replicating this in cochlear implant stimulation is desirable~\cite{McDonnell.07SPIE}.  Furthermore, the concept of `pooling' is also relevant to problems in brain-machine interfaces where recordings need to be made from neural activity. In this context, pooling is similar to what is known as `aggregation'~\cite{Johnson.08}.

\subsection{Analog-to-Digital Converter (ADC) Circuits and Digitized Beamforming}

Continuing advances in digital signal processing technology have led to a trend for the ADC in digital communications networks and sensors to be shifted `as close to the antenna as possible.' A simple signal processing task affected by this trend is in-band random noise reduction via {\em averaging}. If averaging is carried out on sequential samples, it is known in radar or sonar as {\em coherent} or {\em pre-detection} integration. Spatial averaging also results during the process of beamforming. Assuming infinite precision, and finite variance noise, coherently averaging $N$ independently noisy observations of the same signal of course reduces the noise by a factor of $\sqrt{N}$.

However, if each measurement has been digitized first, noise is no longer simply reduced by this factor, and the problem can be abstracted as an SPN. Performance of this {\bf `averaging of digitized signals'} depends on the input SNR, the probability distributions of both the signal and the noise, and the number of bits of the ADC~\cite{McDonnell.07IDC,McDonnell.Icand_book}. This scenario is different from dithered ADCs in that more than one noise source is assumed in an SPN~\cite{McDonnell}, although it can be thought of as similar to a spatial translation of the low-pass filtering properties of oversampling ADCs~\cite{Zozor.05}.  A single ADC where each voltage comparator is an independently noisy node may also be modelled as an SPN~\cite{Oliaei.03,Nguyen.06}.

\subsection{Other Candidates for SPN modelling}

The simple binary SPN first considered in~\cite{Stocks.Mar2000} has provided inspiration for novel reliability schemes in nanoscale electronics~\cite{Martorell.08}. At this scale, SNRs are very small, meaning it may be impossible for traditional noise reduction methods to operate. New methods are required~\cite{Lee.03}, including the possibility of using SPN-like arrays to reduce noise~\cite{Martorell.08}. SPNs may also be useful in studies of complex social networks. The subjective voting example given at the start of this section is a very simple case, and could be extended in numerous directions. We also suggest that any other situation where compressed measurements are averaged, such as some image processing tasks, may be suitable for studying as an SPN.

\section{Stochastic Pooling Networks: Definition}\label{S:Define}

For the current paper we define an SPN as a network without any feedback, co-operation, adaptation, or side-information, although we believe such features could be introduced without altering the basic concept. Otherwise, the following definition is, by design, very general. In Section~\ref{S:Cases} we consider specific examples that illustrate the interesting features of SPNs, and for which we can quantify performance.

\subsection{Essential Features}

We begin by expanding on the three features asserted in Section~\ref{S:Intro} to be essential elements of an SPN.
\begin{description}
  \item[\textbf{Multiple sensors make stochastic observations of a common signal}:] The signal being measured may be any information source\footnote{We use the term `information source' qualitatively, but intend it to refer to some parameter that it is desirable to measure, and that can be modeled as a random variable.}, and in general may be either scalar or vector. By stochastic observations we mean that each measurement is a random variable when conditioned on the common signal. This stochasticity may arise either externally -- e.g. additive or multiplicative random noise -- or internally in a node. Each node's output measurement may be either conditionally independent of each other, or correlated, but not perfectly correlated as then those nodes could be treated as a single node.
\item[\textbf{Each sensor communicates compressed measurements}:] The input to a node may be the information source itself, or a mixture of the information source and random noise. Each node in the network is required to carry out {\em lossy compression} on its input, meaning the node output is a discrete variable with less states than the input, and the original input cannot be exactly recovered from the output. It is permissible for different compression operations in each node.
  \item[\textbf{Measurements are combined by pooling}:] The outputs of each node are combined at a `fusion hub' to form an overall
network output consisting of a single observable measurement of the information source. This combining is characterised by a {\em pooling function}, i.e. an abstraction that accounts for all corruption or processing of the vector of the outputs from an SPN's individual nodes.
\end{description}
Note that we have not yet fully defined some restrictions on the form of {\em pooling function} that makes an SPN interesting to study -- we will address this below.

The first two features listed above are deliberately very general, and it is not particularly useful to define the nodes of an SPN solely in such terms. For example, we have referred to individual node outputs as not being `perfectly correlated.' While in many cases we would like to assume that the conditional outputs of all nodes are independent, we do not want to exclude the possibility of correlation. We also use the general term `lossy compression,' which could mean the application of a complex algorithm. However, it is more likely that very simple compression, e.g. signal quantization such as occurs in an analog-to-digital converter circuit, is required to fit the definition of `pooling.' Hence, particular case studies of SPNs should carefully define the properties of individual network nodes, since although the first two general properties are necessary conditions, the precise nature of the nodes is less important to the overall concept of an SPN. Instead, the primary property that needs careful definition is the third one, i.e. the nature of the pooling function.

\subsection{The Pooling Function}

The pooling function can be thought of as a measurement `fusion hub.' In signal processing applications, it is usually assumed that fusion algorithms are open to {\em engineered design}, and that the whole toolkit of optimal filtering and compression algorithms is available for the design.

{\em Is this the case for an SPN?} Our answer to this question is {\em no}, and we wish to embed this assertion in our definition of an SPN, and more specifically, in a definition of what it means to `pool' in an SPN. This approach is due to the motivation.  We are not interested here in {\em designing} a good or optimal network as a whole, or in designing good or optimal fusion and compression schemes.  Techniques for achieving such tasks abound in the modern communications engineering and signal processing literature and practise.

Instead, we want to address the question of whether good or optimal fusion can emerge {\em naturally} in physical and biological systems. That is, can networks naturally combine -- compressively -- correlated measurements in a way that does not lose any information {\em or} loses negligible information? We discuss more precisely what we mean by `loss of information' in Section~\ref{S:loss}. This motivation means we are interested in modelling the part of the SPN between the points where (i) each node produces a compressed measurement, and (ii) the whole network produces its final output, as  a {\em channel}. That is, this pathway is an environment governed by physical properties and constraints that cannot be controlled by external intervention or adaptation. {\em This means we can define `pooling' in terms of properties of a physical channel.} It is this notion that most contributes to making an SPN, as we have defined it, a generically useful and interesting model.

This leads us to at last state a definition of a pooling network.
\begin{quote}
 {\bf A stochastic pooling network is a network with the property that multiple parallel stochastic and compressed measurements of a common input signal are combined into a single measurement by a physical channel, in such a way that pooling of measurements causes no (or negligible) further loss of information about the network's input signal, when compared with the best performance that could be achieved if all sensor measurements were available.}
\end{quote}
Note that if a network exists such that pooling is highly lossy, then it is excluded from our definition.

Our intention with this definition is to move the emphasis from the precise properties of individual nodes onto the properties of the whole network. This does create a difficulty in precisely defining the pooling function -- i.e., what class of functions should be specifically excluded, and what functions are `simple' enough to emerge naturally from physical properties?   We therefore do not propose such a class of functions here, but use the example of pooling by {\em summation} in Section~\ref{S:Cases}. We can envisage also pooling functions being the maximum or minimum measurement, or the majority vote, and -- as discussed in Section~\ref{S:sufficient} -- any {\em sufficient statistic} created from the individual measurements. However, we have yet to fully explore the possibilities.  Provided a network where the pooling function emerges from physical properties rather than design, and meets the `no or negligible loss of information' property, then it may be an SPN according to our definition.

\subsection{Discussion on Information Loss}\label{S:loss}

So far we have used the term `information' in a qualitative sense only. To concretely define the `nodes are lossy compressors' and `pooling with no or negligible information loss' properties, it is necessary to state the technical sense in which we mean `information.'  There are several possible ways of analyzing the performance of an SPN, depending on the signal processing goal of the network.  Here we use Shannon {\em mutual information}~\cite{Cover2}, although other definitions of information, such as Fisher information, and distance measures for discrimination problems, may be equally appropriate~\cite{Lehmann}. The mutual information between two random variables $\alpha$ and $\beta$ is denoted as $I(\alpha;\beta)$, and has units of {\em bits per sample}, or {\em bits per channel use}, depending on the context.

Suppose an SPN has $N$ nodes, and the information source is the random variable $x$, which may be either discrete or continuously valued. We denote the output of each individual node as the discrete random variables $y_i,~i=1,..,N$, where $y_i$ has cardinality $M_i$, and the overall output of the SPN as the discrete random variable $y$, with cardinality $K$. Without loss of generality, we label each state of $y_i$ by the integers $0,..,{M_i-1}$ and each state of $y$ by the integers $0,..,{K-1}$.  Figure~\ref{f:SPN} shows a block diagram of an SPN as we have defined it.

\subsubsection{Compressive Nodes}

Even though in some cases the stochasticity of a node may be completely internal, for the general formulation we consider that each node's measurements can be corrupted by external noise, $\eta_i$, and recognize that it may happen that $\eta_i$ is always zero. We denote the input to each node as a function of the signal and noise, $w_i=f_i(x,\eta_i)$.  Each node operates on this noisy input to obtain a discretely valued measurement $y_i$, according to a conditional probability mass function (PMF)
\begin{equation}\label{NodeProbs}
    P_{y_i|w_i}(y_i=u|w_i) = g_i(u,w_i)\quad u\in0,..,M_i-1.
\end{equation}
In this paper we define $w$ as a continuous random variable. This is sufficient to guarantee the lossy compression property we require of an SPN, since $M_i$ is a discrete random variable. We note that it may be valid for $w$ to be a discrete random variable, but do not consider this in the current paper. In terms of mutual information, it is always true that $I(x;w_i)\ge I(x;y_i)$, since otherwise the {\em data processing inequality} is violated -- see, e.g., Theorem 2.8.1 in~\cite{Cover2}, and~\cite{McDonnell.03EL}. If $x$ is a continuous random variable the inequality is always strict, since $y_i$ is discrete.

In some of the examples in Section~\ref{S:Models}, it may appear that the output of a node is not a discrete random variable. However we are constructing an abstraction of the true physical properties that takes into account only the macroscopic properties that impact on the production of an output measurement from the pooling hub. For example, in the case of a neuron population, the only property of a neuron that codes information is whether it produces an action potential at a certain time. Hence, the fact that the output voltage is in reality an analogue variable that may be subject to fluctuations, is not material.  Random fluctuations in this variable that may impact on the pooling hub's output can be incorporated into the PMF model of each node.

\subsubsection{No or negligible information loss}

We denote the total number of possible states in the vector of individual node outputs as $M_g:=\prod_{i=1}^NM_i$. Without further compression or compaction, this vector can be represented using $B_g=\sum_{i=1}^N\log_2{(M_i)}$ bits. We are interested in networks for which the {\em pooling function}, $y:=h(y_1,y_2,..,y_N)$, is such that the following two properties hold.
\begin{enumerate}
  \item  The pooling function $h(\cdot)$ is such $K<M_g$, so that the number of bits required to represent $y$ without further compression or compaction is $B_y=\log_2{(K)}<B_g$. Without this property, the network cannot be said to `pool.'
  \item The mutual information between the information source and the vector of observations is either equal to that between the source and the network output,
\begin{equation}\label{Lossless}
    I(x;y_1,y_2,..,y_N)=I(x;y),
\end{equation}
or is such that
\begin{equation}\label{LossLeast}
    I(x;y_1,y_2,..,y_N)=I(x;y)+\epsilon,
\end{equation}
where $\epsilon$ is positive and small compared to $I(x;y)$.  This property states what we mean by `no or negligible loss of information.'
\end{enumerate}
These two properties together mean that the pooling function combines $N$ measurements in a way that reduces the number of raw bits required to code the measurements, but without incurring the cost of reduced mutual information.

We remark that Eq.~(\ref{Lossless}) does not mean `lossless compression' in the sense that the term is usually understood. In computer science and information theory, `lossless compression' describes compression that is deterministically reversible, such as file compression schemes, or Huffman coding~\cite{Cover2}. Data is coded in such a way that it can be stored and transmitted using less bits than the original data but later recovered exactly by an inverse operation.  When Eq.~(\ref{Lossless}) holds in an SPN, the vector $(y_1,..,y_N)$ cannot be recovered from $y$, and yet there is no reduction in mutual information. This is known as redundancy reduction, or data compaction, rather than lossy compression. On the other hand, if Eq.~(\ref{LossLeast}) holds, the pooling function must be a lossy compressor.

Clearly, if the pooling function could be designed or controlled, the fact that it might compress or compact the network's measurements is {\em not} remarkable.  We reiterate that what we are interested in modelling with an SPN is the situation where measurements are combined by `pooling,' so that compression/compaction from $B_g$ to $B_y$ bits, with little reduction in mutual information, {\em emerges solely from the physical properties and constraints of the system}. Although it could be desirable that the pooling function was truly a lossless compressor, this constraint means that the pooling function will almost certainly not be a complicated algorithm. This excludes any networks where the pooling hub is capable of carrying out reversible lossy compression algorithms like Huffman coding, or more sophisticated data compaction algorithms.

\subsection{Discussion on Pooling Functions and Sufficient Statistics}\label{S:sufficient}

A useful perspective is to note that Eq.~(\ref{Lossless}) is equivalent to stating that $y=h(y_1,..,y_N)$ is a {\em sufficient statistic} for the observations $y_1,..,y_N$~\cite{Cover2,Lehmann}. If pooling in an SPN is lossless, then by definition the pooling function must be a sufficient statistic created from the vector of measurements from individual nodes -- see Case Study 1 in Section~\ref{S:Cases}.

On the other hand, if the loss is negligible, as in Eq.~(\ref{LossLeast}), we cannot say anything definitive with regard to sufficient statistics. Nevertheless, allowing the loss of {\em negligible} information is an appealing notion, because in many circumstances small losses are acceptable.  In Case Studies 2 and 3 in Section~\ref{S:Cases}, we use numerical methods to show that particular pooling functions that are not sufficient statistics incur negligible loss of mutual information, meaning the studied networks are SPNs.

\section{Case Studies}\label{S:Cases}

We make the following assumptions for the three case studies presented in this section:
\begin{itemize}
  \item $g_i(\cdot)$ is a scalar input, scalar output function;
  \item the cardinality of $y_i$ is the same for all nodes, i.e.  $M_i=M~\forall~i=1,..,N$;
  \item the information source, $x$, is a discrete time stationary random signal, with no memory, i.e. each sample is {\em iid} with PDF $f_x(\cdot)$, and support $S$;
  \item all analysis is for a pooling function that {\em sums} the individual measurements, i.e. $y=h(y_1,..,y_N)=\sum_{i=1}^Ny_i$;
  \item except for the Poisson nodes in Case Study 1, the stochasticity of each node is due to external continuously valued {\em iid} additive random noise, $\eta_i$, that is uncorrelated with the signal, i.e. $w_i=x+\eta_i$;
\end{itemize}
Illustrative example plots of the mutual information, as a function of input SNR, are presented after discussing each case. In these plots, the input SNR is that at the input to an individual SPN node for Gaussian noise, apart from the Poisson case, where input SNR is undefined, and so plots are scaled relative to the mean of the Bernoulli case. It is straightforward to calculate mutual information by numerical integration for any given distribution of $x$, provided the conditional output distribution is known. Details can be found, for example, in~\cite{McDonnell.07}. In all plots, the signal is assumed to be Gaussian. The choice of Gaussian signal and noise is arbitrary. Other finite variance distributions do not lead to qualitatively different results.

\subsection{Case Study 1: Identical `Poisson-like' nodes}

Here we discuss a class of SPNs where a pooling function that sums the vector of node outputs always results in a sufficient statistic for that vector, and hence summation satisfies Eq.~(\ref{Lossless}).  This case study assumes all network nodes have identical PMFs conditioned on their inputs, i.e. we let $g_i(w_i,u)=g(w_i,u)~\forall~i=1,..,N$ in Eq.~(\ref{NodeProbs}), and denote the cardinality of $y_i$ as $M_i=M~\forall~i$. We remark that although the following analysis applies to continuously valued $y_i$ as well, any such network is not an SPN according to our definition.  It is possible that future generalizations of the SPN concept may relax the condition of discretely valued $y_i$.

We are interested in the case when the conditional PMF of each node's output, {\em conditioned on the input signal}, $x$, is from an exponential family, with a linear sufficient statistic~\cite{Lehmann}. This means functions $a(x)$, $b(x)$ and $Z(x)$ exist such that we can write
\begin{equation}\label{exp_fam}
P_{y_i|x}(y_i=u|x) = \frac{a(u) e^{ b(x) u} }{Z(x)}~\forall~i=1,..,N.
 \end{equation}
When written in this form, $Z(x)$ is known as the partition function, which is very important in statistical physics, since most thermodynamical quantities may be found from it and its derivatives. The {\em linearity} of the sufficient statistic for individual nodes leads to a simple pooling function. Since each of the $N$ SPN nodes are identical, their outputs are conditionally {\em iid}, and their joint conditional PMF can be written as
\begin{equation}\label{joint}
P_{y_1,\ldots,y_N | x}(u_1,\ldots,u_N|x) = \frac{ e^{b(x)y}\prod_{i=1}^Na(u_i) }{Z(x)^N}.
\end{equation}
where $y=\sum_{i=1}^Nu_i$.
It is clear that this joint PMF is also from an exponential distribution, with $y$ as a sufficient statistic for $(y_1,..,y_N)$.

We now discuss two examples of SPNs where the nodes can be described in the form of Eq.~(\ref{exp_fam}), and hence pooling by summation satisfies Eq.~(\ref{Lossless}). First, suppose each node is an identical binary quantizer that compares its input to a threshold level, $\theta$, so that $y_i\in\{0,1\}~\forall~i$. Each node is defined by the PMF
\begin{equation*}
g(u,w_i) = P_{y_i|w_i}(u|w_i) = \left\{\begin{array}{cc}
                          u & w_i \ge \theta \\
                          1-u & w_i < \theta,
                        \end{array}\right.~\quad~u=0,1.
\end{equation*}
Suppose also that the input to each node is $w_i=x+\eta_i$, i.e. the information source is subject to {\em iid} additive random noise.  Upon letting $p(x)=1-F_\eta(\theta-x)$, where $F_\eta(\cdot)$ is the cumulative distribution function of the noise, the conditional PMF for given $x$ is a Bernoulli distribution with parameter $p(x)$ and can be written as
\begin{equation*}
P_{y_i|x}(u|x) = p(x)^u ( 1- p(x) )^{1-u}\quad~u\in\{0,1\}.
\end{equation*}
The Bernoulli distribution is known to be an exponential family, and can be rewritten in the form of Eq.~(\ref{exp_fam}) by letting $a(u)=1$, $b(x) = \log{\left(\frac{p(x)}{1-p(x)}\right)}$, and $Z(x)=\frac{1}{1-p(x)}$.

As a second example we consider inherently stochastic nodes that are conditionally Poisson, with rate $\lambda(w_i)$. The conditional PMF of node $i$ is
\begin{eqnarray*}
P_{y_i|w_i}(u|w_i) = \frac{e^{-\lambda(w_i) } \lambda(w_i)^{u}}{u!}\quad~u=0,1,..,\infty.
\end{eqnarray*}
If there is no external noise, so that $w_i=x~\forall~i$, the conditional PMF given the input is again an exponential family with a linear sufficient statistic, and can be written in the form of Eq.~(\ref{exp_fam}) by letting $a(u)=1/u!$, $b(x) = \log{(\lambda(x))}$ and $Z(x)=\exp{(\lambda(x))}$.

If every node in an SPN has its conditional PMF $g(\cdot)$ described by either of these examples, then it is clear that if the pooling function is $y=\sum_{i=1}^Ny_i$, then pooling satisfies Eq.~(\ref{Lossless}), since $y$ is a sufficient statistic for $(y_1,..,y_N)$.   Figure~\ref{f:BinarySPN_Info} shows the behaviour of the mutual information with increasing input SNR for these examples, in comparison with the mutual information of an example of Case Study 3 below.

Naturally occurring summation of parallel measurements can occur for either example. The binary quantization case occurs in distributed sensor networks that use multiaccess channels~\cite{Li.08}, while the Poisson distribution is widely used to model the production and pooling of action potentials in biological neurons~\cite{Gerstner}.  We have shown that efficient processing of information is an emergent property of such networks.

We remark that it has recently been argued that the summation of the output of two populations of neurons encoding information about one stimulus is Bayes optimal for the estimation of the stimulus, if the likelihood of the population is `Poisson-like'~\cite{Ma.06}. `Poisson like' means that the likelihood is an exponential family with linear sufficient statistics. This is exactly the case in the examples described here.

\subsection{Case Study 2: Identical $M$-ary Quantizing Nodes and {\em iid} Additive Noise}

We now consider a case where the SPN models the `averaging of digitized signals' scenario outlined in Section~\ref{S:Models}. We show that while a summing pooling function is not a sufficient statistic, the reduction in mutual information after pooling by summation is negligible, that is, Eq.~(\ref{LossLeast}) is satisfied.

Here each network node represents an analog-to-digital converter circuit. Each digital sample is assumed to be of an independently noisy version of the same signal sample. We model each circuit as an {\em identical} $M$-ary quantizer that is defined by $M-1$ threshold levels, $(\theta_1,..,\theta_{M-1})$. This means we let $g_i(\cdot)=g(\cdot)~\forall~i$. The vector of individual node outputs has $M_g=M^N$ possible states, but is reduced on pooling to $K=N(M-1)+1$ states, as shown in Fig.~\ref{f:DigitalAveraging}. This example is of particular interest in the context of an SPN being a network with some redundancy, because in the absence of any noise the network has maximum redundancy -- i.e.~$\eta_i=0~\forall~i$ -- and is not an SPN.

It is not useful to write each node in the form of Eq.~(\ref{NodeProbs}), but upon defining $\theta_0:=-\infty$ and $\theta_{M}:=\infty$ we can write
\begin{equation*}
y_i = u~\quad \theta_{u} \le w_i < \theta_{u+1}~\quad~u = 0,..,M-1.
\end{equation*}
For an SPN with $N$ such nodes, each node can be thought of as providing a single {\em iid} multinomial trial. There are $N$ trials in total, each of which can produce any of $M$ outcomes. The probabilities of each outcome are
\begin{equation*}
P_{y_i|x}(u|x) = \left\{\begin{array}{cc}
                         F_\eta(\theta_1-x) & u=0, \\
                         F_\eta(\theta_{u+1}-x)-F_\eta(\theta_u-x) & u=1,..,N-1, \\
                         1-F_\eta(\theta_N-x) & u=M.
                       \end{array}
\right.
\end{equation*}
Just as for the binary quantizer scenario in Case Study 1, it is well known that the PMF of a multinomial trial is from an exponential family, as is the result of $N$ {\em iid} multinomial trials. A sufficient statistic for the outcome of $N$ multinomial trials is the count of how many times each outcome occurs, which can be written as an $M-1$ dimensional vector, since the sum of the counts must add to $N$. If this sufficient statistic is formed, the distribution of the sufficient statistic is given by the multinomial distribution. This distribution has $\left(^{N+M-1}_{M-1}\right)$ states.

In the binary case, a pooling function that forms the sum $y=\sum_{i=1}^Ny_i$ is sufficient, since it is equivalent to a count of the number of $1$s. However for $M>2$, a pooling function that forms the counts of each of $M$ outcomes does not occur naturally in scenarios like the `averaging of digitized signals,' case described in Section~\ref{S:Models}. Of course, there are often many sufficient statistics, and one that may be useful in the context of this SPN is stated in the form of the following theorem.
\begin{theorem}
Let $S_j=\sum_{i=1}^Ny_i^j$. Then $S:=(S_1,S_2,..,S_{M-1})$ is a sufficient statistic for $(y_1,y_2,..,y_N)$. Hence, $I(x;(y_1,y_2,..,y_N)) = I(x;(S_1,S_2,..,S_{M-1}))\ge I(x;y)$.
\end{theorem}
\begin{proof}
Let $q_m(u)$ be an $(M-1)$--th order polynomial with roots being the integers between $0$ and $M-1$, except for $m$, and such that $q_m(m)=1$. This polynomial exists, and can be written as $q_m(u)=\alpha_m \prod_{j=0, j\not=m}^{M-1}(u-j)$ where
$\alpha_m=(\prod_{j=0, j\not=m}^{M-1}(m-j))^{-1}$. This set of polynomials is called the Lagrange basis, and we can write
\begin{eqnarray*}
P_{y_i|x}(u|x) = \prod_{m=0}^{M-1} p_m(x)^{q_m(u)}= \exp{(\sum_{m=0}^{M-1} q_m(u)\log{(p_m(x))})}.
\end{eqnarray*}
It is clear that we have a multiparameter exponential family, with an $M-1$ dimensional sufficient statistic, the vector $(q_0(u),q_1(u),..,q_{M-1}(u))$. However, it is possible to rewrite the PMF in the form
\begin{eqnarray*}
P_{y_i|x}(u|x) = \exp{(\sum_{m=0}^{M-1} b_m(x)u^{m})},
\end{eqnarray*}
and the vector $(u^0,u^1,..,u^{M-1})$ is also a sufficient statistic. Similar to the binomial case in Case Study 1, if there are $N$ {\em iid} multinomial trials, their joint distribution can also be written as an $M-1$ parameter exponential family, where the vector of sufficient statistics is as stated in Theorem 1, after defining $S_j=\sum_{i=1}^Ny_i^j$.\qedhere

\end{proof}
So, provided that $N>M-1$, the sufficient statistic $S$ reduces the number of measurements without incurring a loss in mutual information, since $B_g=N\log_2{(M)}$, while the raw number of bits required to code $S$ is $B_S=(M-1)\log_2{(N+1)}$. We note that this is far less efficient than the count of the frequency of each state.

Like the state counts, the sufficient statistic $S$ may not be a naturally occurring pooling function in an SPN. However, it is feasible that each node's output could be raised to a power before summation. So we may naturally be interested in the case where the pooling function is simply $y=S_j$ for any integer $j$. Obviously the case of $j=1$ corresponds to the `averaging of digitized signals model.'  Verification that this pooling function satisfies Eq.~(\ref{LossLeast}) is given by Figure~\ref{f:TrinarySPN}. The double peak in Figure~\ref{f:TrinarySPN_InfoDiff} is intriguing. It indicates that there is an optimal input SNR close to 10 decibels where the pooling loss approaches zero, for all $N$ considered.  Further investigation of this is left for future work.

\subsection{Case Study 3: Non-identical binary-quantizing nodes}

In Case Studies 1 and 2 we considered only the case where all nodes were identical.  This allowed us to write their conditional PMFs as exponential families, and find sufficient statistics for the vector of $N$ measurements. However, in general, if the network nodes do not all have the same PMF it is not possible to write the joint PMF as an exponential family, and it is difficult to find sufficient statistics. Nevertheless, the traces in Figure~\ref{f:BinarySPN_InfoDiff} demonstrate numerically that pooling by summation can still lead to a negligible reduction in mutual information.

For this case, the small improvement in mutual information that can be had if there is no pooling relies on an assumption that the ordered vector of nodes is available at a `decoder.' In other words, the measurements from each node must somehow be labelled.  It is more natural in the context of SPNs to assume that measurements are combined in a way such that any such labelling is lost. The thought experiment at the start of Section~\ref{S:Models} is an example where there is no labelling. In such a case, there is no way to match events to nodes, and counting the occurrences of each event becomes a sufficient statistic. Under such an assumption for Case Study 3, pooling by summation loses no information.  We have previously refereed to this as a {\bf no labelling} property~\cite{McDonnell.08ACTW}.

\section{Emergent Properties and Unsolved Problems on SPNs}\label{S:Emergent}

In Section~\ref{S:Intro} we stated that SPNs can display some surprising emergent behaviour. One of these is suprathreshold stochastic resonance, which can be seen for Case Studies 1 and 2 in Figures~\ref{f:BinarySPN_Info} and~\ref{f:TrinarySPN_Info}, in agreement with~\cite{Stocks.Mar2000,McDonnell.Icand_book}.  Starting with tiny SNRs, the mutual information increases gradually from zero with increasing input SNR, reaches its peak at an optimal SNR near $0$ decibels, and then decreases as the SNR continues to increase, that is, suprathreshold stochastic resonance occurs. Identical nodes are redundant in the absence of noise, and the mutual information can be no larger than that of a single node. However, when independent noise is present at the input to each node, all nodes contribute to coding the input signal. Noise can be said to provide a benefit relative to the absence of noise. In a sense, suprathreshold stochastic resonance occurs because random noise improves sub-optimal compression. However, in many examples in Section~\ref{S:Models}, there may be no other way to improve on compression of the network as a whole, and allowing an optimal noise level is the only way to enhance performance.

Other emergent behaviour that might not be predicted from the definition of an SPN have been demonstrated for the binary node case. These include the following.
\begin{itemize}
  \item If the nodes in Case Study 3 are optimized for the network as a whole, for small SNRs Case Study 1 is optimal, i.e. all nodes are identical~\cite{McDonnell.06}. For intermediate SNRs, it is optimal for clusters of nodes to be identical, with more unique nodes as the SNR increases. This occurs in a series of bifurcations~\cite{McDonnell.06,McDonnell}.
  \item The mutual information of very noisy SPNs approaches that of analogue Gaussian channels, i.e. $I(x,y)<0.5\log_2{(1+N\mbox{SNR})}$ while near noiseless SPNs are limited by quantization, $I(x,y)<\log_2{(1+N(M-1))}$~\cite{McDonnell.08ACTW}.
  \item Very large SPNs (i.e. $N\rightarrow\infty$) behave like multiplicative noise channels, and optimal reconstruction of the information source depends only on the noise distribution and $N$~\cite{McDonnell.07}.
  \item Optimizing the noise distribution in an SPN is like optimizing a neuron's stimulus-response curve~\cite{McDonnell.07,McDonnell.PRL08}.
\end{itemize}

As stated in Section~\ref{S:Intro}, the aim of this paper is to define SPN, and to illustrate why SPNs are interesting and generically useful in a number of scenarios. There are many areas which remain unexplored, and we simply list some quite general unsolved questions that we believe would be useful to address.

\begin{itemize}
  \item In what scenarios can naturally occurring pooling functions be sufficient statistics, or lead to negligible reduction in mutual information?
  \item Can mathematical approaches predict the clustering of nodes seen in~\cite{McDonnell.06,McDonnell}?
  \item Can our SPN definition be extended to apply to networks with more complex topologies, incorporating aspects like feedback loops, adaptation, and cooperation between nodes, while maintaining the essential qualitative properties?
  \item Can studies of SPNs inspire new designs for electronic systems, similar to~\cite{Martorell.08}?
  \item Can the features and emergent properties of SPNs be observed to occur in-vivo in biological neurons?
\end{itemize}

We hope that presentation of these unsolved problems will stimulate further work and debate on SPNs in the areas of statistical physics, biophysics and electronics design, and that our definition of {\em stochastic pooling network} will eventually evolve and be refined. In summary, the SPN concept may be appropriate for networks where redundancy is useful for achieving noise reduction and simplicity, and where lossy compression is required for efficiency.

\ack
Mark D. McDonnell is funded by an Australian Research Council Post Doctoral Fellowship, DP0770747. N. G. Stocks is funded by EPSRC grants EP/D051894/1 and EP/C523334/1.

\begin{figure}[htbp]
\begin{center}\includegraphics[scale=0.3]{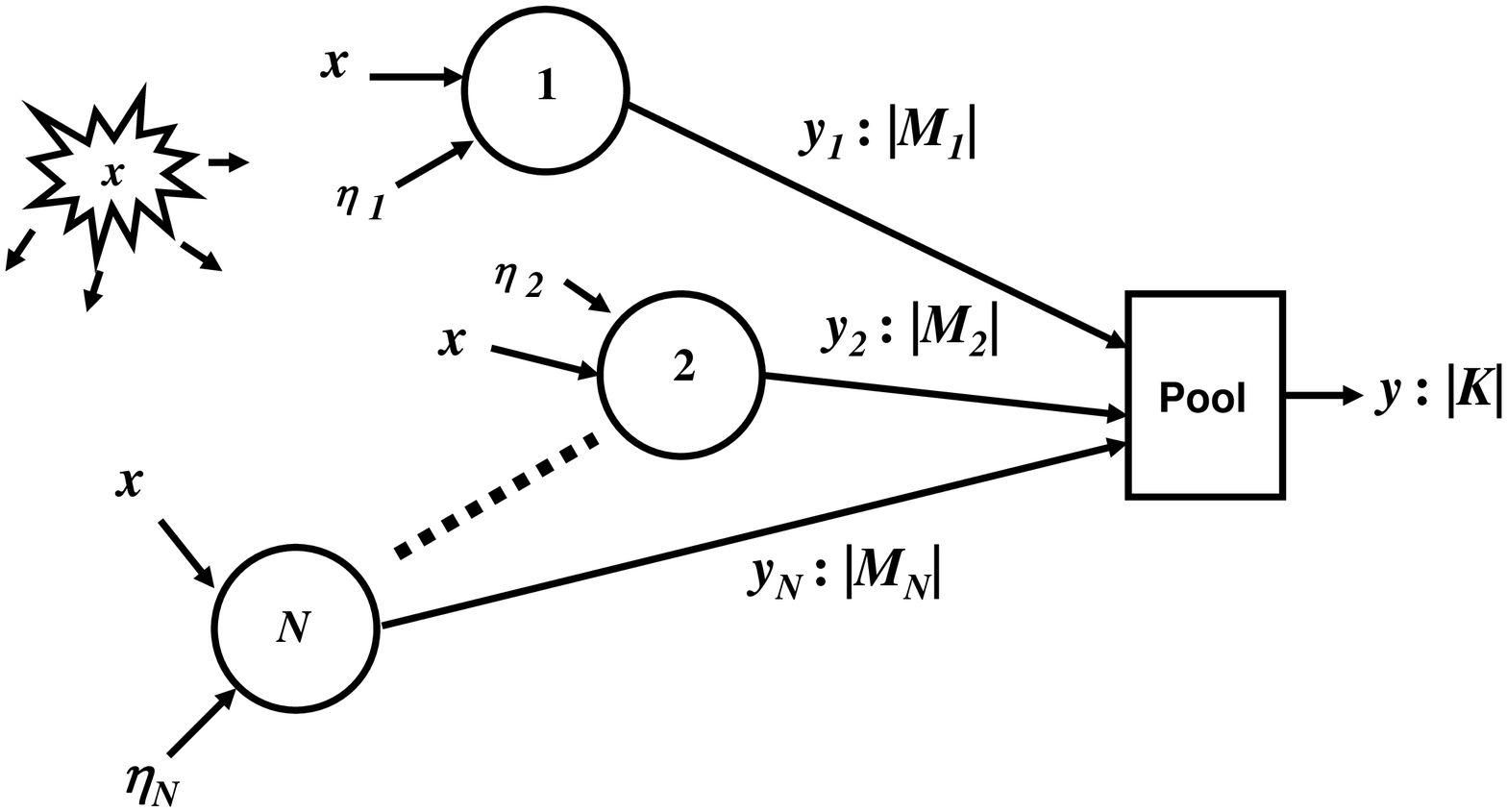}
\caption{An SPN consists of $N$ nodes, each of which operates on noisy versions of the same sample from an information source, the random variable, $x$. Each node's output is stochastic, and is defined for a given signal sample by a conditional probability distribution. This is indicated by $N$ random variables, $\eta_i,~i=1,..,N$. It is permissible for $\eta_i$ to be correlated across nodes. The output of the $i$--th node is an $M_i$ state discrete random variable, , $y_i$, and if $x$ is continuous this means each node is lossy. The overall network `pools' the outputs from each node, in a channel governed by physical properties, to provide an overall network output, $y$, with $K$ states, and is compressed relative to the vector of node outputs.}\label{f:SPN}
\end{center}
\end{figure}

\begin{figure}[htbp]
\begin{center}
{\subfigure[Case Studies 1 and 3]{\includegraphics[scale=0.3]{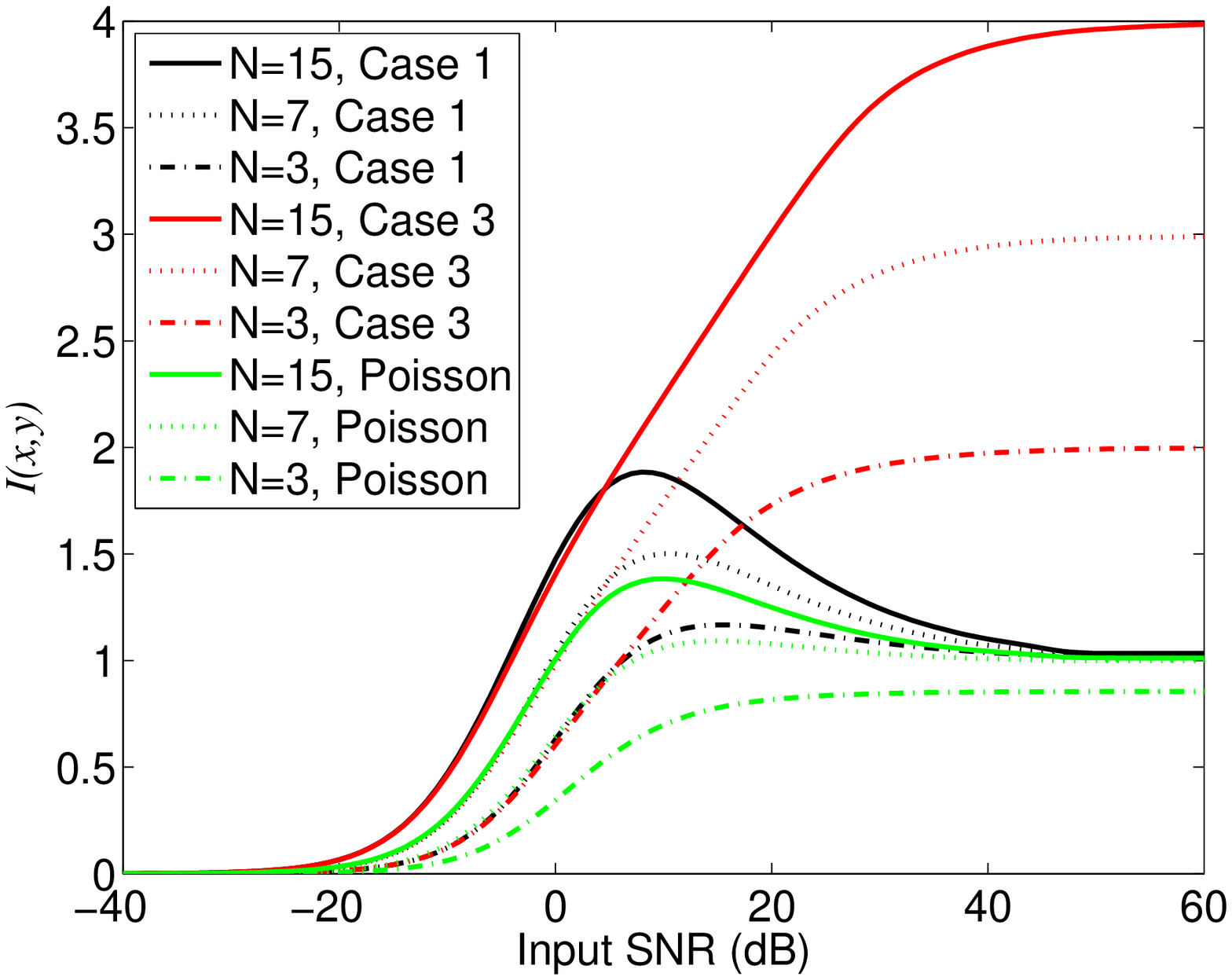}\label{f:BinarySPN_Info}}
\subfigure[Loss in MI for Case study 3]{\includegraphics[scale=0.3]{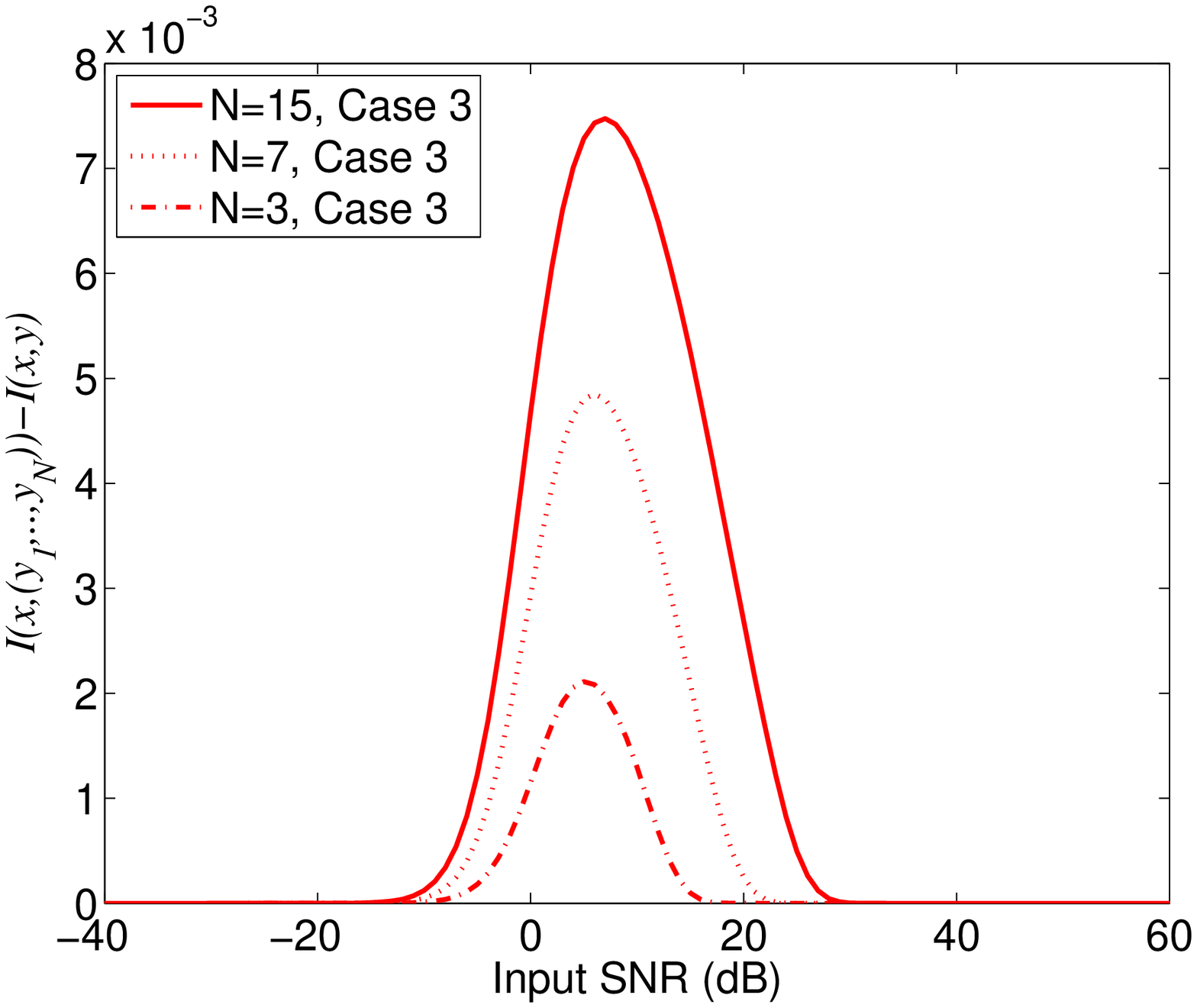}\label{f:BinarySPN_InfoDiff}}}
\caption{(a) Mutual information for SPNs with (i) additive noise and identical binary ($M=2$) quantizing nodes (black traces); (ii) identical Poisson nodes (green traces); and (iii) non-identical noisy binary quantizing nodes (red traces). The Poisson case shown is not plotted against SNR, but is scaled such that the expected value of $y$ is the same as for the binary quantizing case at each SNR. This is achieved by $\lambda(x)=p(x)$. (b) Difference between mutual information before and after pooling by summation, for Case Study 3. The plots for Case Study 3 are for threshold values chosen to optimize the mutual information in the absence of any external noise.}\label{f:BinarySPN}
\end{center}
\end{figure}

\begin{figure}[htbp]
\begin{center}\includegraphics[scale=0.3]{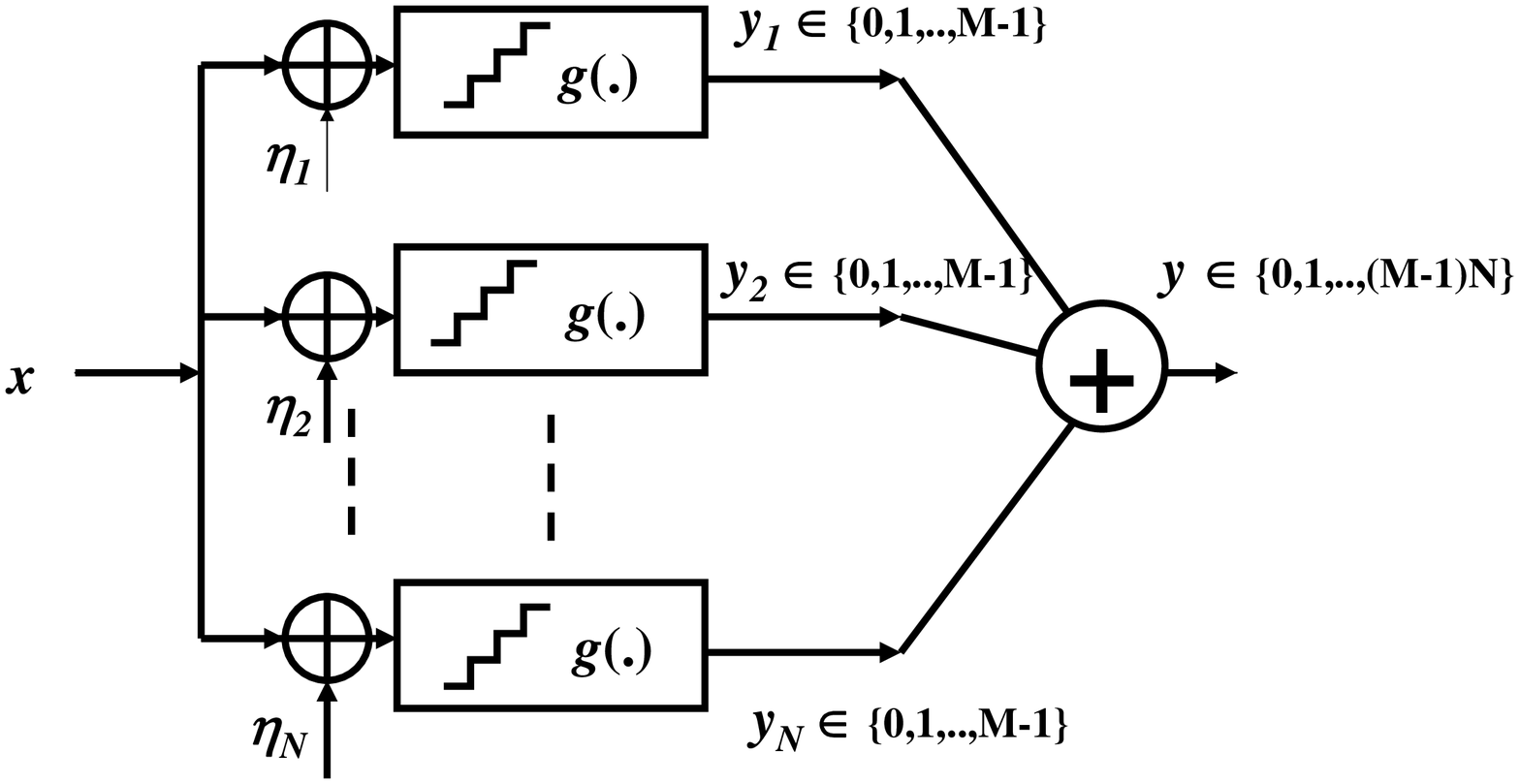}
\caption{SPN for Case Study 2. Each of $N$ nodes is an $M$-ary quantizer. The input to each node is a common random signal corrupted by {\em iid} additive noise. The pooling function simply sums the outputs from each node, to result in a $K=(M-1)N+1$-state discrete random variable. This SPN models the averaging of digitized signals, as discussed in Section~\ref{S:Models}.}\label{f:DigitalAveraging}
\end{center}
\end{figure}

\begin{figure}[htbp]
\begin{center}
{\subfigure[Case Study 2]{\includegraphics[scale=0.3]{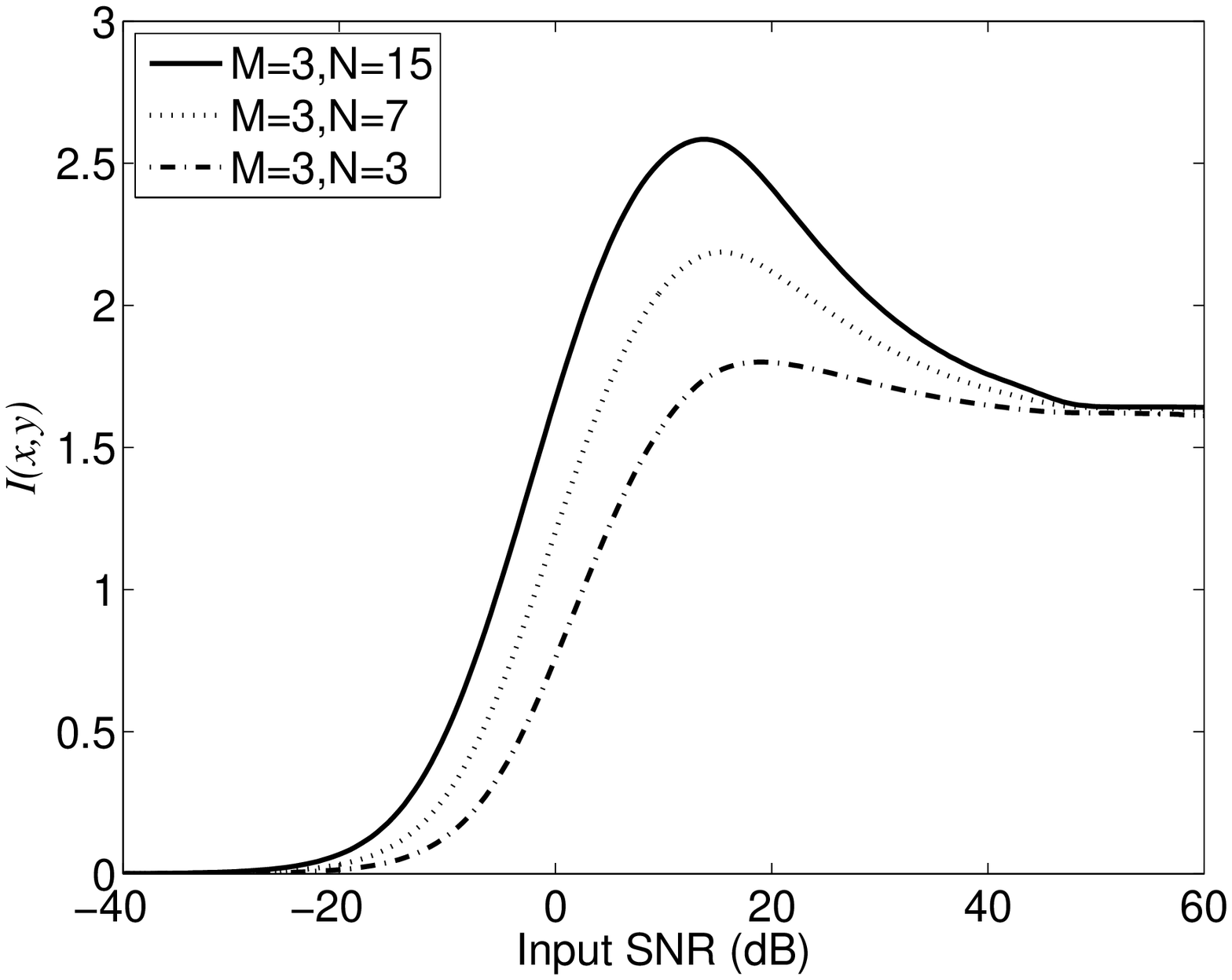}\label{f:TrinarySPN_Info}}
\subfigure[Loss in MI for Case Study 2]{\includegraphics[scale=0.3]{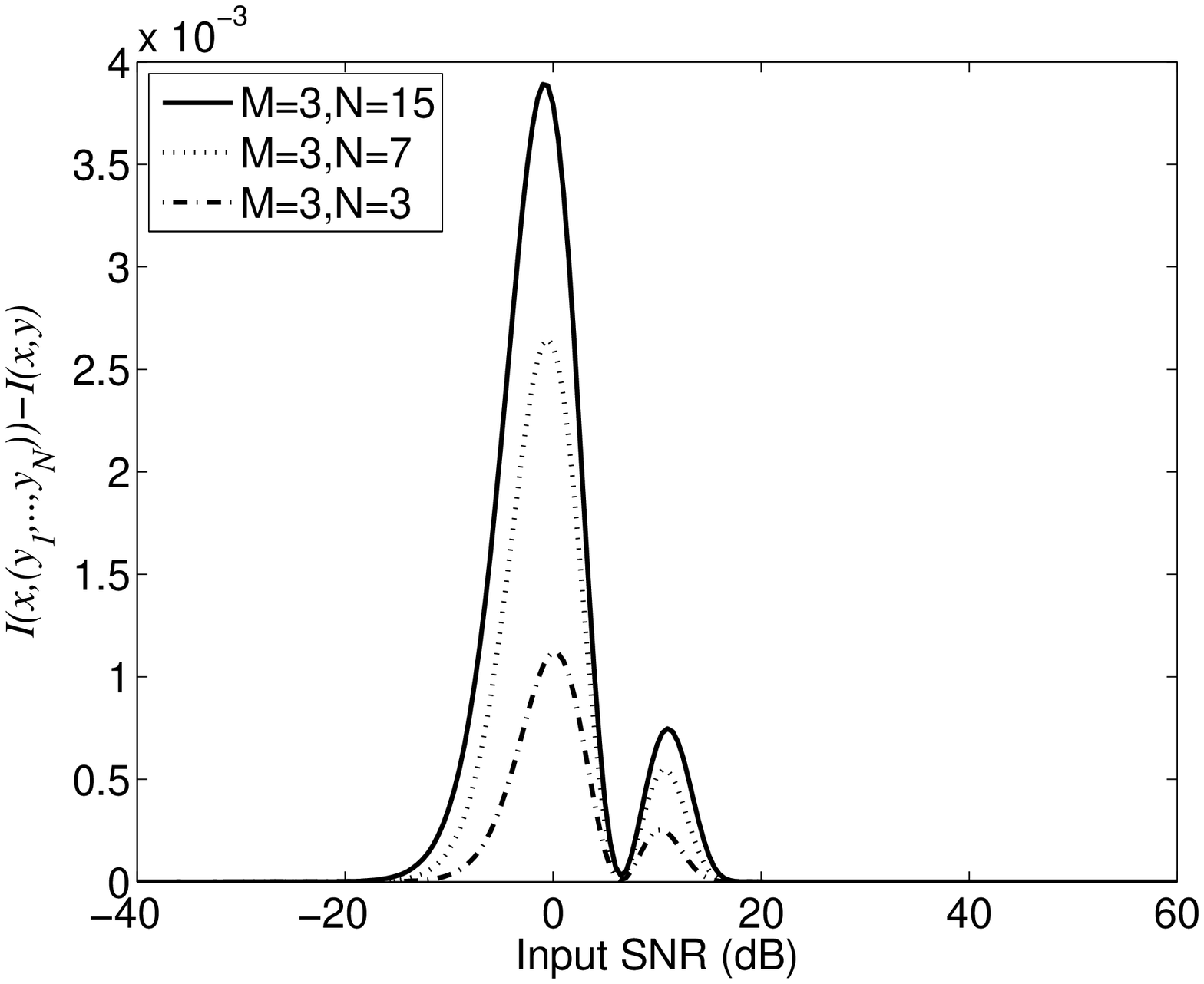}\label{f:TrinarySPN_InfoDiff}}}
\caption{(a) Mutual information for an SPN with additive noise and $N$ identical trinary ($M=3$) quantizing nodes, that each have the same threshold levels $\theta_1$ and $\theta_2$. (b) Difference between mutual information before and after pooling by summation.}\label{f:TrinarySPN}
\end{center}
\end{figure}

\end{document}